\documentclass{article}

\usepackage{amsmath}
\usepackage{algcompatible}
\usepackage{graphicx,psfrag,amsmath,amssymb,epsf}
\usepackage{epsfig,graphics}

\newtheorem{theorem}{Theorem}[section]
\newtheorem{lemma}{Lemma}[section]
\newtheorem{corollary}{Corollary}[section]
\newtheorem{definition}{Definition}[section]

\newtheorem{observation}{Observation}[section]

\newcommand{\qed}{\hfill\hbox{\rlap{$\sqcap$}$\sqcup$}}
\newenvironment{proof}{\noindent \emph{Proof.\,}}{\qed}

\usepackage[linesnumbered,boxed]{algorithm2e}

\title{The Longest Subsequence-Repeated Subsequence Problem}

\author{Manuel Lafond\footnote{Department of Computer Science, Université de Sherbrooke, Sherbrooke, Quebec J1K 2R1, Canada. Email: {\tt manuel.lafond@usherbrooke.ca}.}
\and 
Wenfeng Lai \footnote{College of Computer Science and Technology, Shandong University, Qingdao, China. Email: {\tt 2290892069@qq.com}.}
\and 
Adiesha Liyanage \footnote{Gianforte School of Computing, Montana State University, Bozeman, MT 59717, USA. Email: {\tt adiesha@gmail.com}.}
\and 
Binhai Zhu \footnote{Gianforte School of Computing, Montana State University, Bozeman, MT 59717, USA. Email: {\tt bhz@montana.edu}.} 
}

\date{}

\begin{document}

\maketitle

\begin{abstract}
Motivated by computing duplication patterns in sequences, 
a new fundamental problem called the longest subsequence-repeated subsequence (LSRS) is proposed.
Given a sequence $S$ of length $n$, a letter-repeated subsequence is a subsequence of $S$ in the form of $x_1^{d_1}x_2^{d_2}\cdots x_k^{d_k}$ with $x_i$ a subsequence of $S$, $x_j\neq x_{j+1}$ and $d_i\geq 2$ for all $i$ in $[k]$ and $j$ in $[k-1]$. 
We first present an $O(n^6)$ time algorithm to compute the longest cubic subsequences of all the $O(n^2)$ substrings of $S$, improving the trivial $O(n^7)$ bound. Then, an $O(n^6)$ time algorithm for computing the longest subsequence-repeated subsequence (LSRS) of $S$ is obtained. Finally we focus on two variants of this problem. We first consider the constrained version when $\Sigma$ is unbounded, each letter appears in $S$ at most $d$ times and all the letters in $\Sigma$ must appear in the solution. We show that the problem is NP-hard for $d=4$, via a reduction from a special version of SAT (which is obtained from 3-COLORING). We then show that when each letter appears in $S$ at most $d=3$ times, then the problem is solvable in $O(n^4)$ time.
\end{abstract}

\section{Introduction}

Finding patterns in long sequences is a fundamental problem in string
algorithms, combinatorial pattern matching and computational biology.
In this paper we are interested in long patterns occurring at a global level, which has also been considered previously. One prominent example
is to compute the longest square subsequence of a string $S$ of length $n$, which was solved by Kosowski in $O(n^2)$ time in 2004 \cite{DBLP:conf/spire/Kosowski04}. The bound is conditionally optimal as any
$o(n^{2-\varepsilon})$ solution would lead to a subquadratic bound for the traditional Longest Common Subsequence (LCS) problem, which is not possible unless the SETH conjecture fails \cite{DBLP:conf/focs/BringmannK15}.
Nonetheless, a slight improvement was presented by Tiskin \cite{Tiskin13}; and Inoue et al. recently tried to solve the problem by introducing the parameter $M$ (which is the number of matched pairs in $S$) and $r$ (which is the length of the solution) \cite{DBLP:conf/spire/InoueIB20}.

In biology, it was found by Szostak and Wu as early as in 1980 that gene duplication is the driving force of evolution \cite{uneqcrnature}. There are two kinds of duplications: arbitrary segmental duplications (i.e., an arbitrary segment is selected and pasted at somewhere else) and tandem duplications (i.e., in the form of $X\rightarrow XX$, where $X$ is any segment of the input sequence).
It is known that the former duplications occur frequently in cancer genomes \cite{Ciriello13,CGARN11,Sharp05}. On the other hand, the latter are common under different scenarios; for example, it is known that the tandem duplication of
3 nucleotides {\tt CAG} is closely related to the Huntington disease \cite{Macdonald93}. In addition, tandem duplications can occur at the genome level (acrossing different genes) for certain types of cancer
\cite{Oesper12}. 

As duplication is common in biology, it was not a surprise that in the first sequenced human genome around 3\% of the genetic contents are in the form of tandem repeats \cite{Lander01}.
In 2004, Leupold et al. posed a fundamental question regarding tandem duplications: what is the complexity to compute the minimum tandem duplication distance between two sequences $A$ and $B$ (i.e., the minimum number of tandem duplications to convert $A$ to $B$). In 2020, Lafond et al. answered this open question by proving that this problem is NP-hard for an unbounded alphabet \cite{DBLP:conf/stacs/LafondZZ20}. Later in \cite{Lafond21}, Lafond et al. proved that the problem is NP-hard even if $|\Sigma|\geq 4$ by encoding each letter in the unbounded alphabet proof with a square-free string over a new alphabet of size 4 (modified from Leech's construction \cite{leech57}), which covers the case most relevant with biology, i.e., when $\Sigma=\{{\tt A},{\tt C},{\tt G},{\tt T}\}$ or $\Sigma=\{{\tt A},{\tt C},{\tt G},{\tt U}\}$ \cite{Lafond21}. Independently, Cicalese and Pilati showed that the problem is NP-hard for $|\Sigma|=5$ using a different encoding method \cite{DBLP:conf/iwoca/CicaleseP21}.

Besides duplication, another driving force in evolution is certainly mutation. As a simple example, suppose we have a toy singleton genome
${\tt ACGT}$ (note that a real genome certainly would have a much larger alphabet) and it evolves through two tandem duplications ${\tt ACGT}\cdot {\tt ACGT} \cdot {\tt ACGT}$ then another one on the second ${\tt GTA}$ to have $H={\tt ACGT}\cdot {\tt AC}\cdot {\tt GTA}\cdot {\tt GTA} \cdot {\tt CGT}$. If in $H$ some mutation occurs, e.g., the first $G$ is deleted and the second $G$ is changed to $T$ to have
$H'={\tt ACT}\cdot {\tt AC}\cdot {\tt TTA}\cdot {\tt GTA}\cdot {\tt CGT}$, then it is difficult to retrieve the tandem duplications from $H'$. Motivated by the above applications, Lai et al. \cite{DBLP:conf/cpm/LaiLZZ22} recently proposed the following problem called the {\em Longest Letter-Duplicated Subsequence}: Given a sequence $S$ of length $n$, compute
a longest letter-duplicated subsequence (LLDS) of $S$, i.e., a subsequence of $S$ in the form $x_1^{d_1}x_2^{d_2}\cdots x_k^{d_k}$ with $x_i\in\Sigma$, where $x_j\neq x_{j+1}$ and $d_i\geq 2$ for all $i$ in $[k]$, $j$ in $[k-1]$ and $\sum_{i\in [k]}d_i$ is maximized. A simple linear time algorithm can be obtained to solve LLDS. But some constrained
variation, i.e., all letters in $\Sigma$ must appear in the solution, is shown to be NP-hard.

In this paper, we extend the work by Lai et al. by looking at a more general version of LLDS, namely, the Longest Subsequence-repeated Subsequence (LSRS) problem of $S$, which follows very much the same definition as above
except that each $x_i$ is a subsequence of $S$ (instead of a letter). As a comparison, for the sequence $H'$, one of the optimal LLDS solutions is ${\tt AATTTT}={\tt A}^2{\tt T}^4$ while the LSRS solution is ${\tt ACAC}\cdot {\tt TAGTAG}=({\tt AC})^2({\tt TAG})^2$ which clearly gives more information about the duplication histories.
This motivates us studying LSRS and related problems in this paper.
Let $d$ be the maximum occurrence of any letter in the input
string $S$, with $|S|=n$. 
Let \emph{LSDS+}$(d)$ be the constrained version that all letters in $\Sigma$ must appear in the solution, and the maximum occurrence of
any letter in $S$ is at most $d$.
We summarize the results of this paper as follows. 
\begin{enumerate}
    \item We show that the longest cubic subsequences of all substrings of $S$ can be solved in $O(n^6)$ time, improving the trivial $O(n^7)$ bound.
    \item We show that LSRS can be solved  in $O(n^6)$ time.
    \item When $d\geq 4$, \emph{LSRS+}$(d)$ is NP-complete.
    \item When $d=3$, \emph{LSRS+}$(3)$ can be solved in $O(n^4)$ time.
\end{enumerate}
Note that the parameter $d$, i.e., the maximum duplication number, is practically meaningful in bioinformatics, since whole genome duplication is a rare event in many genomes and the number of duplicates is usually small. For example, it is known that plants have undergone up to three rounds of whole genome duplications, resulting in a number of duplicates bounded by 8 \cite{zheng2009gene}.

It should also be noted that our LSRS and LSRS+ problems seem to be related to the recently studied problems Longest Run Subsequence (LRS) \cite{DBLP:journals/almob/SchrinnerGWSSK21}, which is NP-hard; and Longest (Sub-)Periodic Subsequence \cite{DBLP:journals/corr/abs-2202-07189}, which is polynomially solvable. But these two problems are different from our LSRS and LSRS+ problems. For instance, in an LRS solution a letter can appear in at most one run while in our LSRS and LSRS+ solutions,
say ${\tt ACAC}\cdot {\tt TAGTAG}$ for the input string $H'$, a substring (e.g., {\tt AC}) can appear many times, hence a letter (e.g., {\tt A}) could appear many times but non-consecutively in LSRS and LSRS+ solutions. On the other hand,
in the Longest (Sub-)Periodic Subsequence problem one is very much only looking for the repetition of a single subsequence of the input string, while obviously in our LSRS and LSRS+ problems we need to find the repetitions of multiple subsequences of the input string (e.g., {\tt AC} and {\tt TAG}).

This paper is organized as follows. In Section 2 we give necessary definitions. In Section 3 we give an $O(n^6)$ time algorithm for computing the longest cubic subsequences of all substrings of $S$, as well as the solution for LSRS. In Section 4 we prove that $\emph{LSRS+}(4)$ is NP-hard and then we show that $\emph{LSRS+}(3$) can be solved in polynomial time. We conclude the paper in Section 5.

\section{Preliminaries}

Let $\mathbb{N}$ be the set of natural numbers. For $q\in\mathbb{N}$, we use $[q]$ to represent the set $\{1,2,...,q\}$ and we define $[i,j] = \{i, i+1, \ldots, j\}$.
Throughout this paper, a sequence $S$ is over a finite alphabet $\Sigma$.
We use $S[i]$ to denote the $i$-th letter in $S$ and $S[i..j]$ to denote
the substring of $S$ starting and ending with indices $i$ and $j$ respectively. (Sometimes we also use $(S[i],S[j])$ as an interval representing the substring $S[i..j]$.) With the standard run-length representation, $S$ can be represented as
$y_1^{a_1}y_2^{a_2}\cdots y_q^{a_q}$, with $y_i\in\Sigma,y_j\neq y_{j+1}$ and $a_j\geq 1$, for $i\in[q],j\in[q-1]$. 
Finally, a {\em subsequence} of $S$ is a string obtained by deleting some letters in $S$. Specifically, a {\em square subsequence} of $S$ is a subsequence of $S$ in the form of $X^2$, where $X$ is also a subsequence of $S$; and a
{\em cubic subsequence} of $S$ is a subsequence of $S$ in the form of $X^3$, where $X$ is a also subsequence of $S$. One is certainly interested in the longest ones in both cases.

A subsequence $S'$ of $S$ is a subsequence-repeated subsequence (SRS) of $S$ if it is in
the form of $x_1^{d_1}x_2^{d_2}\cdots x_k^{d_k}$, with $x_i$ being a subsequence of $S$, $x_j\neq x_{j+1}$ and $d_i\geq 2$, for $i\in[k],j\in[k-1]$. We call each $x_i^{d_i}$ in $S'$ a {\em subsequence-repeated block} (SR-block, for short).
For instance, let $S={\tt ACGAGCGCAGCGA}$, then $S_1={\tt AGAG}\cdot {\tt CGCGCG}$, $S_2={\tt ACGACG}\cdot {\tt CGCG}$ and $S_3={\tt ACGACG}\cdot {\tt CACA}$ are multiple solutions for the longest subsequence-duplicated subsequence of $S$, where any maximal substring in $S_i$ separated by $\cdot$ forms a SR-block. 
As a separate note, given this $S$, the longest square subsequence is ${\tt CAGCG\cdot CAGCG}=({\tt CAGCG})^2$ and
the longest cubic subsequence is ${\tt CGACGACGA}=({\tt CGA})^3$.

\section{A Polynomial-time Solution for LSRS}

In this section we proceed to solve the LSRS problem. Firstly, as a subroutine, we need to compute the longest cubic subsequences of all $O(n^2)$ substrings of $S$ in $O(n^6)$ time. 
Assuming that is the case, we have a way to solve LSRS as follows.

\subsection{Solution for the LSRS Problem}

With Kosowski's quadratic solution for the longest square subsequence (even though we could achieve our goal without using it, see Section 3.2) and
our $O(n^6)$ time solution for the longest cubic subsequence  (details to be given in Section 3.2), we solve the
LSRS problem by dynamic programming. We first have the following observation.

\begin{observation}
Suppose that there is an optimal LSRS solution for a given sequence $S$ of length $n$, in the form of $x_1^{d_1} x_2^{d_2} \ldots x_k^{d_k}$. Then it is
possible to decompose it into a generalized SR-subsequence in the form of $y_1^{e_1} y_2^{e_2} \ldots y_p^{e_p}$, where 
\begin{itemize}
    \item $ 2 \leq e_i \leq 3$, for $i\in[p]$,
    \item $p \geq k$,
    \item $y_j$ does not have to be different from $y_{j+1}$, for $j\in[p-1]$.
\end{itemize}
\label{prop:decomp}
\end{observation}

The proof is straightforward:
For any natural number $\ell \geq 2$, we can decompose it as $\ell = \ell_1 + \ell_2 + \ldots + \ell_z \geq 2$, such that $2 \leq \ell_j \leq 3$ for $1 \leq j \leq z$. Consequently, for every $d_i>3$, we could decompose it into a sum
of 2's and 3's. Then, clearly, given a generalized SR-subsequence, we could easily obtain the corresponding SR-subsequence by combining $y_i^{e_i}y_{i+1}^{e_{i+1}}$ when
$y_i=y_{i+1}$.

We now design a dynamic programming algorithm for LSRS.
Let $L(i)$ be the length of the optimal LSRS solution for $S[1..i]$. Let
$Q2[i,j]$ and $Q3[i,j]$ store the longest square and cubic subsequences
of $S[i..j]$ respectively. The recurrence for $L(i)$ is as follows.
\begin{align*}
L(0) & = 0, \\
L(1) & = 0, \\
L(i) & = \max 
\begin{cases}
L(j) + Q2[j+1,i], & j < i-1  \\
L(j) + Q3[j+1,i], & j < i-2 
\end{cases}
\end{align*}

Computing all the cells $Q2[j,k]$ takes $O(n^4)$ time as there are $O(n^2)$ cells and each can be computed using Kosowski's algorithm in quadratic time. (As we will show right after Theorem 1, the $O(n^4)$ time bound can also be obtained without using
Kosowski's algorithm.) Computing all $Q3[j,k]$ takes $O(n^7)$ time: there are $O(n^2)$ cells,
each can be computed in $O(n^5)$ time using the only known brute-force solution. However, in the next subsection we show that
the longest cubic subsequences of all substrings of $S$, i.e., all $Q3[j,k]$ can be computed in $O(n^6)$ time.
Therefore, after $Q2[-,-]$ and $Q3[-,-]$ are all computed, it takes
$O(n^2)$ time to update and fill the whole table $L(-)$.
The value of the optimal LSRS solution for $S$ can be found in $L(n)$. 
Consequently, we have a running time of $O(n^6)$. To make the solution complete, we next show the algorithm for computing the longest cubic subsequences of all substrings of $S$.

\subsection{An $O(n^6)$ Time Bound for the Longest Cubic Subsequences of All Substrings of the Input String}

First of all, notice that an $O(n^5)$ time brute-force solution for the longest cubic subsequence problem is trivial: just enumerate in $O(n^2)$ time all the cuts cutting $S$ into three substrings, 
and then compute the longest common subsequence over this triple of substrings in $O(n^3)$ time. The longest of all would give us the solution.
Then, to compute all $Q3[j,k]$ it takes $O(n^7)$ time since there are $O(n^2)$ cells.
To improve this bound, a different idea is needed.

The idea is that when one computes the longest common subsequence
of three sequences $A,B$ and $C$, one would use dynamic programming
to compute, for each triple of $i,j,k$, the longest common subsequence of $A[1..i],B[1..j]$ and $C[1..k]$. When $i,j$ are fixed this dynamic programming algorithm can in fact compute the longest common subsequences of $A[1..i],B[1..j]$ and all $C[1..k']$, with $1\leq k'\leq k$. Therefore, by enumerating $i$ and $j$, in $O(n^2\cdot n^3)=O(n^5)$ time, we can compute all longest cubic subsequences of a prefix $A\cdot B\cdot C$ of $S$. To compute
the longest cubic subsequences of all substrings of $S$, it suffices to run the above algorithm on every suffix of $S$. Hence, in $O(n)\cdot O(n^5)=O(n^6)$ time we can compute the longest cubic subsequences of all substrings of $S$.

\begin{theorem}
The longest cubic subsequences of all substrings of an input string of length $n$ can be computed in $O(n^6)$ time and $O(n^3)$ space.
\end{theorem}

Note that we can use this idea to compute the longest square subsequences for all substrings of the input string $S$ in $O(n^4)$ time, without using Kosowski's
algorithm at all. In this case, using the standard dynamic programming for computing the longest common subsequence of $A$ and $B$, we compute all longest square
subsequences of a prefix $A\cdot B$ of the input sequence $S$
in $O(n^3)$ time. Then we run this algorithm on all suffix of $S$, giving a total running time of $O(n^4)$. In this process,
there is no need to use Kosowski's algorithm.

Finally, together with the algorithm in Section 3.1, we have the following theorem.

\begin{theorem}
The longest subsequence-repeated subsequence problem can be solved in $O(n^6)$ time.
\end{theorem}

\section{The Variants of LSRS}

In this section, we focus on the following variations of the LSRS problem.

\begin{definition}
\textbf{\emph{Constrained Longest Subsequence-Repeated Subsequence}}\\ ($\emph{LSRS}+$ for short)

{\bf Input}: A sequence $S$ with length $n$ over an alphabet $\Sigma$ and an integer $\ell$.

{\bf Question}: Does $S$ contain a subsequence-repeated subsequence $S'$ with length at least $\ell$ such that all letters in $\Sigma$ appear in $S'$?
\end{definition}

\begin{definition}
\textbf{\emph{Feasibility Testing}} (FT for short)

{\bf Input}: A sequence $S$ with length $n$ over an alphabet $\Sigma$.

{\bf Question}: Does $S$ contain a subsequence-repeated subsequence $S''$ such that all letters in $\Sigma$ appear in $S''$?
\end{definition}

For LSRS+ we are really interested in the optimization version, i.e., to maximize $\ell$. Note that, though looking similar, FT and the decision version of LSRS+ are different: if there is no feasible solution for
FT, certainly there is no solution for LSRS+; but even if there is a feasible
solution for FT, computing an optimal solution for LSRS+ could still be non-trivial.

Finally, let $d$ be the maximum number of times a letter in $\Sigma$ appears in $S$. Then, we can represent the corresponding versions for LSRS+ and FT as
\emph{LSRS+}$(d)$ and $FT(d)$ respectively.

It turns out that (the decision version of) \emph{LSRS+}$(d)$ and $FT(d)$ are both NP-complete when $d\geq 4$, while when $d=3$ both \emph{LSRS+}$(3)$ and $FT(3)$ can be solved in $O(n^4)$ time.
We present the details below.

\subsection{LSRS+(4) is NP-hard}

We first show that $(3^+,1,2^-)$-SAT is NP-complete; in this version of SAT all variables appear positively in 3-CNF clauses (i.e., clauses containing exactly 3 positive literals) and each variable appears exactly once in total in these 3-CNF clauses; moreover, the negation of the variables appear in 2-CNF clauses (i.e., clauses containing 2 negative literals), possibly many times. A {\em valid} truth assignment for an $(3^+,1,2^-)$-SAT instance $\phi$ is one which makes $\phi$ true; moreover, each 3-CNF clause has exactly one true literal.

A folklore reduction was discussed in the internet at some point; here we give a formal sketch of the proof.

\begin{theorem}
$(3^+,1,2^-)$-SAT is NP-complete.
\label{3color}
\end{theorem}

\begin{proof}
As the problem is easily seen to be in NP, let us focus more on the reduction from 3-COLORING. In 3-COLORING, given a graph $G=(V,E)$, one needs to assign one of the 3 colors to each of the vertex $u\in V$ such that for any edge $(u,v)\in E$, $u$ and $v$ are given different colors.

For each vertex $u$, we use $u_1,u_2$ and $u_3$ to denote the 3 colors, then, obviously, we have the 3-CNF clause $(u_1\vee u_2\vee u_3)$. Therefore, the positive 3-CNF formulae are
$$C^+=\bigwedge_{u\in V}(u_1\vee u_2 \vee u_3).$$

We have 2 kinds of 2-CNF clauses. First, for each $u\in V$, we have a type-1 2-CNF clause which demands that one cannot select two colors $i$ and $j$ for $u$ at the same time:
$$\overline{u_i\wedge u_j}=(\bar{u}_i\vee \bar{u}_j),$$ for $1\leq i\neq j\leq 3$. Then, for each edge $(u,v)\in E$, we have a type-2 2-CNF clause which demands that $u$ and $v$ cannot have the same color $i$:
$$\overline{u_i\wedge v_i}=(\bar{u}_i\vee \bar{v}_i),$$
for $i=1,2,3$.

Let $C^-$ be the conjunction of these 2-CNF clauses. Then $\phi=C^+\wedge C^-$, and it is clear that $G$ has a 3-coloring if and only if $\phi$ has a valid truth assignment. The reduction obviously takes linear time. Hence the theorem is proven.
\end{proof}

\begin{figure}
    \centering
    \includegraphics[width=0.35\textwidth]{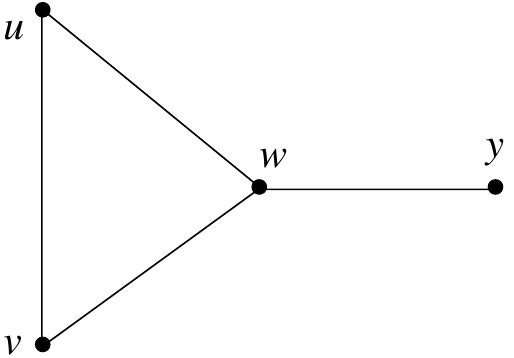}
    \caption{An illustration of the proof of Theorem~\ref{3color}. In this case, $C^+=(u_1\vee u_2\vee u_3)\wedge (v_1\vee v_2\vee v_3)\wedge (w_1\vee w_2\vee w_3)\wedge (y_1\vee y_2\vee y_3)$.}
    \label{fig:3color}
\end{figure}

In Fig.~\ref{fig:3color}, we show an example for the proof of Theorem~\ref{3color}. The example will be used in the following
paragraphs.

We next reduce $(3^+,1,2^-)$-SAT to FT(4). Let the input $\phi$ for $(3^+,1,2^-)$-SAT be constructed directly from a 3-COLORING instance; moreover, let $\phi$ have $3n$ variables $x_1,x_2,\cdots$, $x_{3n}$ and $m$ 2-CNF clauses. We label its 3-CNF clauses as $F^+_1,F^+_2,\cdots,F^+_n$ and its 2-CNF clauses as
$F^-_1,F^-_2,\cdots,F^-_m$. 
(For the example in Fig.~\ref{fig:3color}, we can take as $u_1,u_2,u_3,\cdots,y_1,y_2,y_3$ alphabetically as $x_1,x_2,\cdots,x_{12}$.) 

In the example in Fig.~\ref{fig:3color}, the type-1 2-CNF clauses on $u$ are
$$(\bar{u}_1\vee \bar{u}_2)\wedge (\bar{u}_1\vee \bar{u}_3) \wedge (\bar{u}_2\vee \bar{u}_3)=F^-_1\wedge F^-_2\wedge F^-_3,$$ 
the type-2 2-CNF clauses on edge $(u,v)$ are
$$(\bar{u}_1\vee \bar{v}_1)\wedge (\bar{u}_2\vee \bar{v}_2) \wedge(\bar{u}_3\vee \bar{v}_3)=F^-_4\wedge F^-_5\wedge F^-_6,$$
and the type-2 clauses on edge $(u,w)$ are
$$(\bar{u}_1\vee \bar{w}_1)\wedge (\bar{u}_2\vee \bar{w}_2)\wedge (\bar{u}_3\vee \bar{w}_3)=F^-_7\wedge F^-_8\wedge F^-_9.$$

For each variable $x_i$, let
$L(x_i)$ be the list of type-1 2-CNF clauses containing $\bar{x}_i$ followed with the list of type-2 2-CNF clauses containing $\bar{x}_i$, each repeating twice consecutively. (For the example in Fig.~\ref{fig:3color}, $L(x_1)=L(u_1)=F^-_1F^-_1F^-_2F^-_2\cdot F^-_4F^-_4F^-_7F^-_7$.) For each $F^+_i$ we also define three unique letters $1_i,2_i$ and $3_i$. Hence, the alphabet we use to construct the final
sequence $H$ is $$\Sigma=\{F^-_j|j\in[m]\} \cup \{1_i,2_i,3_i|i\in[n]\} \cup \{g_k,g'_k|k\in[n-1]\},$$
where $g_k$ and $g'_k$ are used as separators.

Let $F^+_i=(x_{i,1}\vee x_{i,2} \vee x_{i,3})$. We construct
$$H_i=2_i\cdot L(x_{i,1})\cdot 1_i2_i1_i\cdot L(x_{i,2})\cdot 2_i \cdot L(x_{i,3})\cdot 3_i2_i3_i,$$ 
where $2_i$ and all the 2-CNF clauses each appear 4 times in $H_i$, while $1_i$ and $3_i$ each appears twice. Finally we construct a sequence $H$ as
$$H=H_1\cdot g_1g'_1g_1g'_1\cdot H_2\cdot g_2g'_2g_2g'_2\cdot H_3 \cdots\cdot g_{n-1}g'_{n-1}g_{n-1}g'_{n-1}\cdot H_n,$$
where $g_j$ and $g'_j$ each appears twice.
We claim that $\phi$ has a valid truth assignment if and only if $H$ induces a feasible SRS which contains all $1_i,2_i,3_i$ ($i=1..n$), all $F^-_j$ ($j=1..m$) and all $g_{k}g'_{k}g_{k}g'_{k}$ ($k=1..n-1$).

The forward direction, i.e., when $\phi$ has a valid truth assignment, is straightforward. In this case, suppose exactly one of $x_{i,1}$, $x_{i,2}$ and $x_{i,3}$ (say $x_{i,j}$, $1\leq j\leq 3$) is assigned TRUE, then we delete $L(x_{i,j})$ for $j=1,2$ or $3$ 
in $H_i$. Finally we delete some $1_i,2_i$ and $3_i$ to obtain a feasible solution $H'_i$ as follows.
\begin{enumerate}
\item If $j=1$, we have $H'_i=2_i1_i2_i1_i \cdot L(x_{i,2})\cdot L(x_{i,3})\cdot 3_i3_i$.
\item If $j=2$, we have $H'_i=L(x_{i,1})\cdot 1_i2_i1_i2_i \cdot L(x_{i,3})\cdot 3_i3_i$.
\item If $j=3$, we have $H'_i=L(x_{i,1})\cdot 1_i1_i \cdot L(x_{i,2})\cdot  2_i3_i2_i3_i$.
\end{enumerate}
It is noted that exactly one $L(x_{i,j})$ is deleted in all three cases. (We focus on $j=1$ next.) Hence, the deleted letters (2CNF clauses in the form of $F^-_k$) in $L(x_{i,1})$ would still appear in the claimed feasible solution, even after the deletion of $L(x_{i,1})$. For example, if $F^-_k$ is type-1 which contains $\bar{x}_{i,1}$ and $\bar{x}_{i,2}$, then $F^-_kF^-_k$ must appear in $L(x_{i,2})$, which is not deleted.
Similarly, if $F^-_k$ is type-2 which contains $\bar{x}_{i,1}$ and $\bar{x}_{\ell,1}$, where $(x_i,x_\ell)$ is an edge in the graph $G$, then $F^-_kF^-_k$ appears in $L(x_{\ell,1})$ which must be
kept --- if $L(x_{\ell,1})$ were also deleted, it would imply that both
$x_i$ and $x_\ell$ are colored with color-1.
Therefore, all the 2CNF clauses appear in the claimed feasible solution.

The reverse direction is slightly more tricky. We first show the following lemma.
\begin{lemma}
If $H$ admits a feasible (SRS) solution, then exactly two (non-empty subsequences) of $L(x_{i,1}), L(x_{i,2})$ and $L(x_{i,3})$ appear in the feasible solution $H'$ (or, exactly one of the three is completely deleted from $H$).
\end{lemma}

\begin{proof}
Due to the type-1 2-CNF clauses constructed on $x_{i,1}$, $x_{i,2}$ and $x_{i,3}$, if exactly one (non-empty subsequence) of the three lists $L(x_{i,1})$, $L(x_{i,2})$ and $L(x_{i,3})$, say $L(x_{i,1})$, appears in a feasible SRS solution $H'$, then the type-1 2-CNF clause involving $x_{i,2}$ and $x_{i,3}$ would be missing in $H'$, contradicting the definition of a feasible solution. 

On the other hand, due to the construction of $H_i$, we claim that one cannot
leave all three (non-empty subsequences) of the lists $L(x_{i,1})$, $L(x_{i,2})$ and $L(x_{i,3})$ in a feasible SRS solution $H'$ (we only need to look at the part obtained from $H_i$, e.g., $H'_i$). 
Suppose it is not the case, we consider six cases: (1) if $1_i1_i$ and $3_i3_i$ are substrings of $H'_i$, (2) if $2_iX\cdot 2_iX$ is a substring of $H'_i$ with the first $X$ being a substring in $L(x_{i,1})$ and the second $X$ being a substring in $L(x_{i,2})$, (3) if $2_iX\cdot 2_iX$ is a substring of $H'_i$ with the first $X$ being a substring in $L(x_{i,2})$ and the second $X$ being a substring in $L(x_{i,3})$, (4) if $X2_i\cdot X2_i$ is a substring of $H'_i$ with the first $X$ being a substring in $L(x_{i,1})$ and the second $X$ being a substring in $L(x_{i,2})$,
(5) if $X2_i\cdot X2_i$ is a substring of $H'_i$ with the first $X$ being a substring in $L(x_{i,2})$ and the second $X$ being a substring in $L(x_{i,3})$, and
(6) $XX$ is a substring of $H'_i$ where the first $X$ is a subsequence of $L(x_{i,1})$ or $L(x_{i,1})\cdot L(x_{i,2})$ or $L(x_{i,1})\cdot L(x_{i,2})\cdot L(x_{i,3})$. We claim that all these six cases are not possible:
\begin{itemize}
\item Case 1: It is impossible as $2_i$'s would not appear in a feasible solution.
\item Cases 2-4: It is impossible as $1_i$'s would not appear in a feasible solution.
\item Case 5: It is impossible as $3_i$'s would not appear in a feasible solution.
\item Case 6: It is impossible as $1_i$'s and $2_i$'s would not appear in a feasible solution.
\end{itemize}
\end{proof}

Let $L'(x_{i,j})$ be a non-empty subsequence of $L(x_{i,j})$. With the above lemma, the reverse direction can be proved as follows. 
\begin{enumerate}
\item If $H'_i=2_i1_i2_i1_i \cdot L'(x_{i,2})\cdot L'(x_{i,3})\cdot 3_i3_i$, then assign $x_{i,1}\leftarrow {TRUE}$, $x_{i,2}\leftarrow {FALSE}$, $x_{i,3}\leftarrow {FALSE}$.
\item If $H'_i=L'(x_{i,1})\cdot 1_i2_i1_i2_i \cdot L'(x_{i,3})\cdot 3_i3_i$,
then assign $x_{i,1}\leftarrow {FALSE}$, $x_{i,2}\leftarrow {TRUE}$, $x_{i,3}\leftarrow {FALSE}$.
\item If $H'_i=L'(x_{i,1})\cdot 1_i1_i \cdot L'(x_{i,2})\cdot  2_i3_i2_i3_i$, then assign $x_{i,1}\leftarrow {FALSE}$, $x_{i,2}\leftarrow {FALSE}$, $x_{i,3}\leftarrow {TRUE}$.
\end{enumerate}
Clearly, this gives a valid truth assignment for $\phi$ --- as all the 2-CNF clauses $(\bar{x}_{i,j}\vee \bar{x}_{k,\ell})$ must appear in $L'(x_{i,j})$ or $L'(x_{k,\ell})$ in $H'$, and at least one of $x_{i,j}$ and $x_{k,\ell}$ is assigned {FALSE}. We thus have the following theorem.

\begin{theorem}
FT(4) is NP-complete.
\end{theorem}
 
Since FT(4) is NP-complete, the optimization problem $\emph{LSRS+}(4)$ is certainly NP-hard.

\begin{corollary}
The optimization version of $\emph{LSRS+}(4)$ is NP-hard.
\end{corollary}

Note that the above proof implies that the optimization version of $\emph{LSRS+}(4)$ does not admit any polynomial-time approximation (regardless of the approximation) as any such approximate solution would
form a feasible solution for FT(4). In fact, using a similar argument as
in \cite{DBLP:conf/cpm/LaiLZZ22}, even finding a good bi-criteria approximation, i.e., approximating the optimal length as well as the maximum number of letters covered, for
$\emph{LSRS+}(4)$ is not possible (unless P=NP). On the other hand, we show next that
$\emph{LSRS+}(3)$ is polynomially solvable.

\subsection{LSRS+(3) is Polynomially Solvable}

We now try to solve LSRS+(3), where the input is a sequence $S$ of length $n$ where each letter appears at most three times and at least twice. As a matter of fact,
an optimal solution must be in the form of $x_1^{d_1}\cdots x_k^{d_k}$, where $x_i$ is a subsequence of $S$ and $d_i\in\{2,3\}$ for $i\in [k]$. Throughout this subsection we assume that all letters in $S$ appear at least twice
and at most three times --- if a letter appears only once in $S$ then there is no solution for the corresponding LSRS+(3) instance.

Our idea is again dynamic programming, based on the above observation that in an optimal solution each SR-block is either a square or a cube. Define 6 tables, $S_{2}[i,j], C_2[i,j], S_{3}[i,j]$, $C_3[i,j]$, $L[i,j]$ and $C[i,j]$, with $1\leq i < j \leq n$. The first 4 are only used to initialize $L[i,j]$'s. 

\begin{itemize}
\item $C_2[i,j]$ is the set of letters that all appear at least twice in $S[i..j]$ if at least one letter appears exactly twice in $S[i,j]$; otherwise $C_2[i,j]$ is empty.

\item $S_2[i,j]$ is the length of a longest square subsequence in $S[i,j]$ containing all the letters in $C_2[i,j]$.
If such a local feasible solution does not exist, set $S_2[i,j]\leftarrow -1$; otherwise, we say that this local feasible solution is {\em 2-feasible}.
 
\item $C_3[i,j]$ is the set of letters that all appear three times in $S[i..j]$ if no letter in $S[i,j]$ appears exactly twice; otherwise $C_3[i,j]$ is empty.

\item $S_3[i,j]$ is the length of a longest cubic subsequence (if exists) in $S[i,j]$ containing all the letters in $C_3[i,j]$; otherwise $S_3[i,j]$ is the length of a longest square subsequence containing all the letters in $C_3[i,j]$. 
If such a local feasible solution (cube or square) does not exist, set
$S_3[i,j]\leftarrow -1$; otherwise, we say that this local feasible solution is {\em 3-feasible} or {\em 2-feasible} respectively (depending on whether the local solution is cubic or square).

\item $C[i,j]$ is the set of letters appearing at least twice in $S[i..j]$. 
The $C[i,j]$'s can be computed, each with a linear scan,
in a total of $O(n^3)$ time. $C[i,j]$'s are only used to enforce the coverage condition.

\item $L[i,j]$ is the length of a feasible solution of $S[i..j]$ which covers all the letters in $C[i,j]$. 
\end{itemize}

The initial values of $L[i,j]$, for $i<j$, can be set as follows.

\begin{align*}
L[i,j] & = 
\begin{cases}
S_3[i,j] &  \text{if~} S_3[i,j]>0  \\
S_2[i,j] & \text{else if~} S_2[i,j]>0 \\
-1 & \text{otherwise.}
\end{cases}
\end{align*}

Note that the initialized solution might not be final. For example,
$S[1..9]=ababbcacc$, then $S_3[1,9]=S_2[1,9]=-1$ and $L[1,9]=-1$, i.e., there is no local 3-feasible (cubic) or 2-feasible (square) solution that covers all the letters in $C[1..9]$. But obviously an optimal solution, i.e., $(ab)^2c^3$, exists. Hence we need to proceed to update $L[i,j]$.

Then we update the general case for $L[i,j]$ recursively as follows. This is done bottom-up, ordered by the ascending length of $S[i..j]$.

\begin{align*}
    L[i,j]&=
    \begin{cases}
    \max\{\max_{i<k<j}\{L[i,k]+L[k+1,j]\},L[i,j]\} & \text{if~} L[i,k]>0, L[k+1,j]>0\text{~and~} \\
    & C[i,j]==C[i,k]\cup C[k+1,j]\\
    -1 & \text{otherwise.}
    \end{cases}
\end{align*}

Two examples can be used to illustrate the update step.
In the first example, $S[1..9]=abacbabcc$, $L[1,9]$ is initially
assigned with $S_3[1,9]=6$ (which corresponds to $(abc)^2$). After the update step $L[1,9]=8$ (i.e., corresponds to $(ab)^3(cc)$).
In the second example, $S[1..9]=abacabccb$, $L[1,9]$ is also
initialized with $S_3[1,9]=6$ (which corresponds $(abc)^2$). After the
update step $L[1,9]=6$ (which corresponds to $(ab)^2(cc)$ or $(abc)^2$ --- in the actual implementation, there is no need to perform an explicit update for this example; but such a piece of information cannot be known before the update step, as shown in the first example).

Note that the condition $C[i,j] == C[i,k]\cup C[k+1,j],$
is to ensure that $L[i,j]$ is updated only when
all the letters in $C[i,j]$ are covered.
Then, the maximum of $L[i,k]+L[k+1,j]$, if greater than $L[i,j]$, replaces (the previous) $L[i,j]$.

An optimal solution is computed if $L[1,n]>0$, and 
its solution value is stored in $L[1,n]$. Clearly, with an additional table, one can easily retrieve such an optimal solution, if exists.


Regarding the correctness of our algorithm, we have several simple lemmas.
\begin{lemma}
$C_2[i,j]\cap C_3[i,j]=\emptyset$.    
\end{lemma}
\begin{proof}
This is obvious, as, by definition, $C_2[i,j]$ is non-empty only when there is a letter appearing exactly twice in $S[i,j]$. On the other hand, $C_3[i,j]$ is non-empty when there is no letter appearing exactly twice in $S[i,j]$. The two conditions are complementary.
\end{proof}
Regarding $S_3[i,j]$, the following lemma says that if a 3-feasible solution does not exist then any of those 2-feasible solutions could be stored.
\begin{lemma}
If a 3-feasible solution for $S_3[i,j]$ does not exist, then any 2-feasible solution for $S[i,j]$ can be stored (without changing the optimal solution value).
\end{lemma}

\begin{proof}
By definition, $C_3[i,j]$ contains all the letters in $S[i,j]$ which appear exactly three times; moreover, there is no letter
$x$ which appears exactly twice in $S[i,j]$. Hence, if a letter
$y$ appears exactly once in $S[i,j]$ it would never appear as
a local feasible solution for $S_3[i,j]$. 

Therefore, if a 3-feasible solution does not exist, by definition, we would consider only a 2-feasible solution for $S[i..j]$ which
covers all the letters in $C_3[i,j]$. The length of such a 2-feasible solution is exactly
$2\cdot |C_3[i,j]|$.
\end{proof}

An example for this lemma is $S[1..6]={\tt baabab}$. There is no 3-feasible solution. On the other hand, either {\tt abab} or {\tt baba} would make a valid 2-feasible solution, to be stored in $S_3[1,6]$. On the other hand, {\tt aabb} could also make a final solution via the update of $L[i,j]$, made of two 2-feasible solutions {\tt aa} and {\tt bb}, for $S_2[1,3]$ and $S_2[4,6]$ respectively. But {\tt aabb} itself is not considered as a 2-feasible solution in, say, $S[1..6]$.

\begin{theorem}
Given a string $S$ of length $n$, where each letter appears at most three times, the problem of LSRS+(3) can be solved in $O(n^4)$ time.
\end{theorem}

\begin{proof}
Regarding the correctness of the update of $L[i,j]$, it is noted that when $L[i,j]$ is updated by $L[i,k]+L[k+1,j]$ the condition that $C[i,j]==C[i,k]\cup C[k+1,j]$ is enforced, i.e., the maximum length of $L[i,j]$ is achieved on the condition that all letters in $S[i,j]$ appear in the corresponding solution. Note also that
$C[i,k]\cap C[k+1,j]=\emptyset$ as each letter can appear at most
three times in $S[i,j]$, if a letter appears two or three times in
$S[i,k]$ it cannot appear two or three times in $S[k+1,j]$ again --- by the definition of $C_2[-,-]$ and $C_3[-,-]$.

The cost of the algorithm is dominated by filling and updating $S_2[i,j]$ and $S_3[i,j]$ as the update of $L[i,j]$'s takes $O(n^3)$ time. $S_3[-,-]$ has
$O(n^2)$ cells, each can be initially filled in $O(n)$ time as a longest cubic subsequence containing all the letters in $C_3[i,j]$, if exists, must be unique --- all these letters appear exactly three times hence we just cut the sequence of these letters in three segments, each with length $|C_3[i,j]|$, to have a 3-feasible solution (if it is). Checking whether all the letters in $C_3[i,j]$ are covered by this cubic subsequence corresponding to $S_3[i,j]$ takes linear time.
If a 3-feasible solution for $S_3[i,j]$ does not exist, then we need to find a longest square subsequence of $S[i,j]$ and check if it has a length $2\cdot |C_3[i,j]|$. (Remember that if both of these two tests fail then we need to update $S_3[i,j]$ as $-1$.) Hence the total time for filling and updating all $S_3[i,j]$ is $O(n^4)$ --- there are $O(n^2)$ cells each can be filled in $O(n^2)$ time using Kosowski's algorithm (or using the method outlined right after Theorem 1 without using Kosowski's algorithm) in the worst case. Similarly, the filling of $S_2[i,j]$ takes $O(n^4)$ time as well.
\end{proof}

Clearly, LSRS+(3) has a solution (i.e., $L[1,n]>0$) if and only if FT(3) has a feasible solution.

\section{Concluding Remarks}
In this paper, we present an $O(n^6)$ time algorithm for computing the longest cubic subsequences of all substrings of an input string $S$ of length $n$. On top of that, we compute the longest subsequence-repeated subsequence of $S$ in $O(n^6)$ time as well. We then consider the constrained longest subsequence-repeated subsequence (LSRS+) and the corresponding feasibility testing (FT) problems in this paper, where all letters in the alphabet must occur in the solutions. We parameterize the problems with $d$, which is the maximum number of times a letter appears in the input sequence. For convenience, we summarize the results one more time in Table 1. Obviously, the most prominent open problem is to decide if it is possible to compute the longest cubic subsequence
in $o(n^5)$ time. Note that besides the trivial brute-force method, using our alignment graph a different $O(n^5)$ time algorithm can be obtained (see Section 3.2).

\begin{table}
	\centering
\caption{Summary of results on LSRS+ and FT.} 
	\begin{tabular}{|c|c|c|c|}
\hline $d$ & \emph{LSRS+}$(d)$ & $FT(d)$&Approximability~of~\emph{LLDS+}$(d)$\\
\hline $d\geq 4$ & NP-hard & NP-complete & No~Approximation \\
\hline $d=3$ & P & P & Not~Applicable\\
\hline
	\end{tabular}
\end{table}

\section*{Acknowledgments}
We thank anonymous reviewers for some insightful comments.


\end{document}